\documentclass[runningheads]{llncs}
\let\oldendproof\endproof
\renewcommand\endproof{~\hfill\qed\oldendproof}

\usepackage{graphicx}
\graphicspath{{figures/}}
\usepackage{amsmath,amssymb}
\usepackage{enumerate}

\usepackage[utf8]{inputenc}
\usepackage[T1]{fontenc}

\usepackage{thmtools, thm-restate}
\usepackage[hidelinks]{hyperref}
\usepackage[nameinlink]{cleveref}
\usepackage{csquotes}
\usepackage{subcaption}
\usepackage{microtype}

\crefname{conjecture}{Conjecture}{Conjectures}
\crefname{proposition}{Proposition}{Propositions}
\crefname{lemma}{Lemma}{Lemmata}
\crefname{theorem}{Theorem}{Theorems}
\crefname{section}{Section}{Sections}
\crefname{appendix}{Appendix}{Appendices}
\crefname{figure}{Fig.}{Figs.} 
\Crefname{figure}{Figure}{Figures} 
\newcommand{\pe}{\mathcal{P}}     
\newcommand{\bigO}{\mathcal{O}}

\newcommand{\sets}{\ell}
\newcommand{\ccm}{\operatorname{ccm}} 
\newcommand{\cmk}{\operatorname{cm}} 
\newcommand{\force}{\operatorname{force}}
\newcommand{\cforce}{\operatorname{cforce}}
\newcommand{\e}{{\operatorname{e}}} 
\newcommand{\id}{\operatorname{id}}

\newcounter{lemmathreesets}

\renewcommand{\orcidID}[1]{\href{https://orcid.org/#1}{\includegraphics[scale=.03]{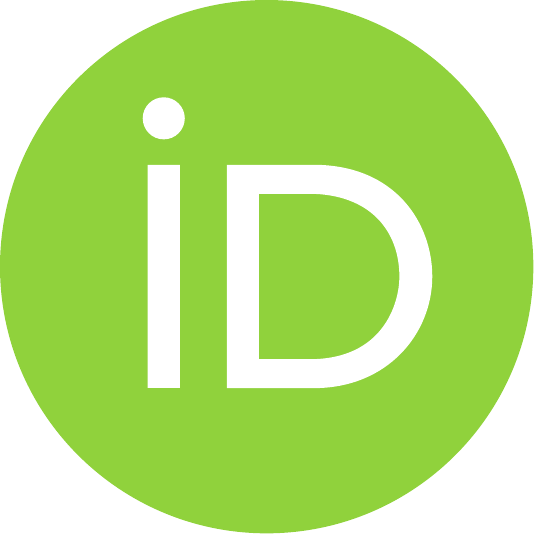}}}

\begin{document}
\title{On Compatible Matchings\thanks{%
A.A.\ funded by the Marie Sk{\l}odowska-Curie grant agreement No 754411.
Z.M.\ partially funded by Wittgenstein Prize, Austrian Science Fund (FWF), grant no.\ Z 342-N31.
I.P., D.P., and B.V.\ partially supported by FWF within the collaborative DACH project \emph{Arrangements and Drawings} as FWF project \mbox{I 3340-N35}. 
A.P.\ supported by a Schr\"odinger fellowship of the FWF: J-3847-N35.
J.T.\ partially supported by ERC Start grant no.\ (279307: Graph Games), FWF grant no.\ P23499-N23 and S11407-N23 (RiSE). 
}}
\author{Oswin~Aichholzer\inst{1}\orcidID{0000-0002-2364-0583} \and
Alan~Arroyo\inst{2}\orcidID{0000-0003-2401-8670} \and
Zuzana~Mas\'arov\'a\inst{2}\orcidID{0000-0002-6660-1322} \and
Irene~Parada\inst{3}\orcidID{0000-0003-3147-0083} \and
Daniel~Perz\inst{1}\orcidID{0000-0002-6557-2355} \and
Alexander~Pilz\inst{1}\orcidID{0000-0002-6059-1821} \and
Josef~Tkadlec\inst{2}\orcidID{0000-0002-1097-9684} \and
Birgit~Vogtenhuber\inst{1}\orcidID{0000-0002-7166-4467}}

\authorrunning{Aichholzer, Arroyo, Mas\'arov\'a, Parada, Perz, Pilz, Tkadlec, and Vogtenhuber} 
\institute{Institute of Software Technology, Graz University of Technology, Austria
\email{\{oaich,daperz,apilz,bvogt\}@ist.tugraz.at}
\and IST Austria \email{\{alanmarcelo.arroyoguevara,zuzana.masarova,josef.tkadlec\}@ist.ac.at}
\and TU Eindhoven, The Netherlands \email{i.m.de.parada.munoz@tue.nl}}
\maketitle    
\begin{abstract} 
	A matching is compatible to two or more labeled point sets of size $n$ with labels $\{1,\dots,n\}$ if its straight-line drawing on each of these point sets is crossing-free. 
We study the maximum number of edges in a matching compatible to two or more labeled point sets in general position in the plane. 
We show that for any two labeled convex sets of $n$ points there exists a compatible matching with $\lfloor \sqrt {2n}\rfloor$ edges.
More generally, for any $\sets$ labeled point sets we construct compatible matchings of size $\Omega(n^{1/\sets})$.
As a corresponding upper bound, we use probabilistic arguments to show that for any $\sets$ given sets of $n$ points there exists a labeling of each set such that the largest compatible matching has $\bigO(n^{2/(\sets+1)})$ edges.
Finally, we show that $\Theta(\log n)$ copies of any set of $n$ points are necessary and sufficient for the existence of a labeling such that any compatible matching consists only of a single edge.
%
%
\end{abstract}
\section{Introduction}
For plane drawings of geometric graphs, the term \emph{compatible} is used in two rather different interpretations. In the first variant, {\it two} plane drawings of geometric graphs are embedded on the {\it same} set $P$ of points. They are called compatible (to each other with respect to $P$) if their union is plane (see e.g.~\cite{abdgh09,ist-dcgm-13}). Note that this is different to simultaneous planar graph embedding, as it is required not only that the two graphs are plane, but also that their union is crossing-free.

In the second setting, which is the one that we will consider in this work, {\it one} planar graph $G$ is drawn straight-line  
on {\it two or more} labeled point sets (with the same label set). We say that $G$ is compatible to the point sets if the drawing of $G$ is plane for each of them 
(where each vertex of $G$ is mapped to a unique label and thereby identified with a unique point of each point set).
Note that the labelings of the point sets can be predefined or part of the solution. As an example, we mention the compatible triangulation conjecture~\cite{AICHHOLZER20033}:
{\em For any two sets $P_1$ and $P_2$ with the same number of points and the same number of points on the boundary of the convex hull, there is a labeling of the two sets such that there exists a triangulation which is compatible to both sets, $P_1$ and $P_2$.}

\subsection*{Motivation and related work}
The study of the type of compatibility considered in this work (the second type from above) is motivated by 
applications in morphing~\cite{morphing0,GD_morphing,morphing1}, 
	2D shape animation~\cite{4815232},
	or cartography~\cite{s_jtatm_87}.

Compatible triangulations were first introduced by Saalfeld~\cite{s_jtatm_87} for labeled point sets. 
He pointed out that for pairs of labeled point sets, compatible triangulations do not always exist and studied the construction of compatible triangulations using (possibly exponentially many) 
Steiner points.
Aronov et al.~\cite{ass_octosp_93} and Babikov et al.~\cite{DBLP:conf/cccg/BabikovSW97} showed that $O(n^2)$ Steiner points are always sufficient, while Pach et al.~\cite{DBLP:journals/algorithmica/PachSS96} showed that $\Omega(n^2)$ Steiner points are sometimes necessary.
Aronov et al.~\cite{ass_octosp_93} also showed that for two labeled polygons, the existence of a compatible triangulation without Steiner points can be determined in polynomial time. The computational complexity question for labeled point sets or polygons with holes is still open. For polygons with holes, Lubiw and Mondal~\cite{LUBIW202097} showed NP-hardnes of deciding the existence of a compatible triangulation with at most $k$ Steiner points.

The compatible triangulation conjecture states that -- in contrast to labeled point sets -- two unlabeled point sets (in general position and with the same number of extreme points) can always be compatibly triangulated without using Steiner points. 
To date, the conjecture has only been proven for point sets with at most three interior points~\cite{AICHHOLZER20033}.
Krasser~\cite{k_ktep_99} showed that more than two point sets cannot always be compatibly triangulated.
Danciger et al.~\cite{DANCIGER2006195} considered compatibly triangulating two or more unlabeled point sets by using few Steiner points.

Concerning compatible paths, Hui and Schaefer~\cite{DBLP:conf/isaac/HuiS04} showed that it is NP-hard to decide
whether two labeled point sets admit a compatible 
spanning path. Arseneva et al.~\cite{DBLP:conf/cccg/ArsenevaBBCCIJL18} presented efficient algorithms for finding a monotone compatible spanning path, or a compatible spanning path inside simple polygons (if they exist).
Czyzowicz et al.~\cite{Paths96} showed that any two convex labeled point sets admit a compatible path of length at least $\sqrt{2n}$ and also presented an $O(n^2\log n)$ algorithm to find such a path. 
In a similar direction, results from Czabarka and Wang~\cite{EScylic19} imply a lower bound of $(\sqrt{n-2}+2)/2$ on the length of the longest cycle compatible to two convex point sets.

In this paper we will focus on compatible matchings. To the best of our knowledge, previous results on (geometric) matchings study only compatibility of the first type, that is, where two matchings are embedded on the same point set. A well-studied question in this setting is whether any two perfect matchings can be transformed into each other by a sequence of steps such that at every step the intermediate graph is a perfect matching and the union of any two consecutive matchings is plane. Aichholzer at al.~\cite{abdgh09} proved that such a sequence of at most $O(\log n)$ steps always exists.
Questions of whether any matching of a given point set can be transformed into any other and how many steps it takes (that is, the connectivity of and the distance in the so-called {\em reconfiguration graph} of matchings, as well as its other properties) have been investigated also for matchings on bicolored point sets and for edge-disjoint compatible matchings,
see for example~\cite{abhpv-ltdbm2018,abls2015,ist-dcgm-13}.

\begin{figure}[tbh]
	\centering
	\begin{minipage}[t]{.43\textwidth}
		\centering
	\includegraphics[scale=0.45, page=2]{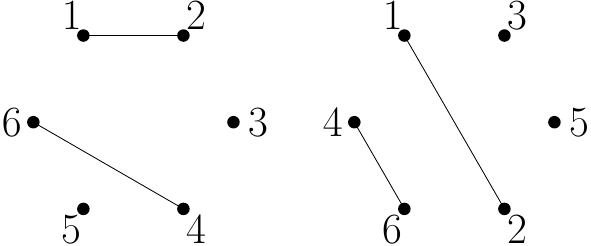}
	\subcaption{}
	\label{fig:simple-a}
	\end{minipage}
\hfill
	\begin{minipage}[t]{.42\textwidth}
	\centering
	\includegraphics[scale=0.45, page=3]{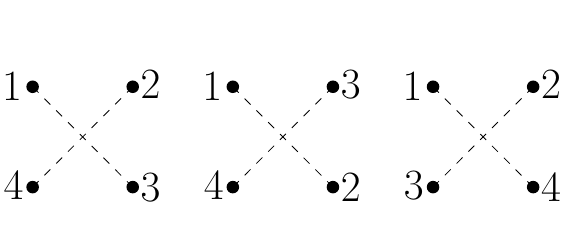}
	\subcaption{}
	\label{fig:simple-b}
	\end{minipage}
	\caption{(a) There is no perfect matching compatible to the two labeled sets. 
	(b)~Any possible pair of matching edges crosses in exactly one of the three sets. }
	\label{fig:simple}
\end{figure} 

\subsection*{Our results}
We study the second type of compatibility for matchings on two or more point sets. 
This is a setting for which no previous comprehensive theory appears to exist.
Throughout this paper, we denote unlabeled point sets with $P$ and labeled point sets with $\pe$ (both mostly with added indices to distinguish between different point sets).

We start by considering convex point sets: given two unlabeled convex point sets $P_1$, $P_2$, both with $n$ points, 
we study the largest guaranteed size $\ccm(n)$ of a compatible matching  
across all pairs of labelings of $P_1$ and $P_2$.
More formally, $\ccm(n)$ is the minimum over all pairs of labelings of the maximum compatible matching size for the accordingly labeled $n$-point set pairs. 
The largest compatible matching for two labeled point sets is not necessarily perfect, see~\cref{fig:simple-a}.
In \cref{sec:2convex}, we present upper and lower bounds on $\ccm(n)$. 
In particular, for any~$n$ that is a multiple of 10, we construct two labeled convex sets $\pe_1,\pe_2$ of $n$ points each, for which the largest compatible matching has $2n/5$ edges. 
Using probabilistic arguments, we obtain an upper bound of $\ccm(n)=\bigO(n^{2/3})$. 
For the lower bound, we show that for any pair of labeled convex point sets $\pe_1$, $\pe_2$ there exists a compatible matching consisting of $\lfloor \sqrt {2n}\rfloor$ edges. This implies that  $\ccm(n)=\Omega(\sqrt n)$.

We further extend our study to consider $\sets$ point sets in general position 
instead of just two point sets in convex position. 
Given $\sets$ unlabeled sets $P_1,\dots,P_{\sets}$, each consisting of $n$ points in general position,  
we denote by $\cmk(n;P_1,\dots,P_{\sets})$ the largest guaranteed size of a compatible matching 
across all $\sets$-tuples of labelings of $P_1,\dots,P_{\sets}$. 
We remark that the size $n$ of the point sets is included in the notation only for the sake of clarity (since our bounds depend on $n$).
In \cref{sec:general} we give bounds on $\cmk(n;P_1,\dots,P_{\sets})$ for any sets $P_1,\dots,P_{\sets}$ of $n$ points in general position. 
Building on the ideas of the proofs for two convex sets, 
we show that $\cmk(n;P_1,\dots,P_{\sets})=\bigO(n^{2/(\sets+1)})$ and that $\cmk(n;P_1,\dots,P_{\sets})=\Omega(n^{1/\sets})$. 

Finally, we investigate the question of how many labeled copies of a given unlabeled point set  
are needed so that the largest compatible matching consists of a single edge. 
Already for four points in convex position, three different sets are needed (and sufficient, see~\cref{fig:simple-b}).
In~\cref{sec:force} we prove that for any given set of $n\geq 5$ points in general position,
$\Theta(\log n)$ copies of it are necessary and sufficient for the existence of labelings forcing that the largest compatible matching consists of a single edge. 

For brevity, a plane matching that consists of $k$ edges is called a \emph{$k$-matching}.

\section{Two convex sets}
\label{sec:2convex}
Throughout this section we consider two labeled convex point sets $\pe_1$, $\pe_2$ consisting of $n$ points each. 
Without loss of generality we assume that $\pe_1$ is labeled $(1,2,\dots,n)$ in clockwise order and that $\pe_2$ is labeled $(\pi(1),\pi(2),\dots,\pi(n))$ in clockwise order for some permutation $\pi\colon[n]\to[n]$.
Note that for convex point sets, the compatible matching question is a purely combinatorial one, in the sense that it only depends on the two cyclic orders of $(1,2,\dots,n)$ given by the labelings rather than on the concrete positions of the points.
Using this fact, we determined $\ccm(n)$ for small values of $n$ by computing the largest compatible matching for each possible pair of labelings. 
The results of those computations are listed in Table~\ref{table:ccm}.

\begin{table}
\center
\begin{tabular}{ r | r | r | r | r | r | r | r | r | r | r | r | r | r }
    $n$ & \hspace{0.1cm} 4 &\hspace{0.1cm} 5 & \hspace{0.1cm} 6 & \hspace{0.1cm} 7 & \hspace{0.1cm} 8 & \hspace{0.1cm} 9 & \hspace{0.05cm} 10 & \hspace{0.05cm} 11 & \hspace{0.05cm} 12 & \hspace{0.05cm} 13 & \hspace{0.05cm} 14 & \hspace{0.05cm} 15 & \hspace{0.05cm} 16 \\ \hline
    $\ccm(n)$ & 2 & 2 & 2 & 3 & 3 & 4 & 4 & 4 & 5 & 5 & 6 & 6 & 6 \\
\end{tabular}
    \smallskip 
\caption{The values of $\ccm(n)$ for $n=4, \dots, 16$.} \label{table:ccm}
\end{table}

In the following, we present lower and upper bounds on the largest guaranteed size $\ccm(n)$ of a compatible matching of any two such sets. 
Starting with lower bounds, we present four pairwise incomparable results (\Cref{thm:2sets-lb}),
each of them giving rise to a polynomial-time algorithm for constructing a compatible $k$-matching with $k=\Omega(\sqrt n)$ edges.
The results are ordered by the size of the obtained compatible $k$-matching, where the last one gives the best lower bound for $\ccm(n)$, while the three other results yield compatible matchings of special structure.
The second result can be generalized to any number of (not necessarily convex) sets (\Cref{thm:lb}). 
We remark that the results in~\cite{Paths96} imply a lower bound of $\sqrt{2n}/2$ for the size of the largest compatible matching, which is weaker than the fourth result.

Before stating the theorem, we introduce the notion of a shape of a matching on a convex point set which, informally stated, captures ``how the matching looks''.
Consider a labeled point set $\pe$ and a plane matching $M$ on $\pe$. 
Let $\pe^M \subseteq \pe$ be the points of $\pe$ that are incident to an edge of $M$. 
The \emph{shape} of $M$ is the combinatorial embedding\footnote{The combinatorial embedding fixes the cyclic order of incident edges for each vertex.} of the union of $M$ and the boundary of the convex hull of $\pe^M$. 
Further, $M$ is called \emph{non-nested} if its shape is a cycle, that is, all edges of $M$ lie on the boundary of the convex hull of $\pe^M$. 
Note that the shape of $M$ also determines the number of its edges (even though some or all of the edges might be ``hidden'' in the boundary of the convex hull of $\pe^M$).
We say that two matchings have the \emph{same shape}, if their shapes are identical, possibly up to a reflection.

\begin{theorem}[Lower bound for two convex sets]
	\label{thm:2sets-lb}
For any two labeled convex sets $\pe_1$, $\pe_2$ of $n$ points each, it holds that:
\begin{enumerate}[{\bf(i)}]
  \item If $n\ge (2k-2)^2+2$ then for any shape of a $k$-matching there exists a compatible $k$-matching having that shape in both $\pe_1$ and $\pe_2$.
  \item If $n\ge k^2+2k-1$ then any maximal compatible matching consists of at least $k$ edges.
  \item If $n\ge k^2+k$ then there exists a compatible $k$-matching that is non-nested in both $\pe_1$ and $\pe_2$.
  \item If $n\ge \frac12k^2+k$ then there exists a compatible $k$-matching.
\end{enumerate}
\end{theorem}

\renewcommand\endproof{\oldendproof}
\begin{proof} \hfill
\begin{enumerate}[{\bf(i)}]
\item By the circular Erd\H os-Szekeres Theorem~\cite{EScylic19}, the permutation $\pi$ contains a monotone subsequence $\sigma$ having length $2k$. The sequence $S=\{x_i | i\in\sigma\}$ of points whose labels belong to $\sigma$ has the same cyclic order in both sets $\pe_1$, $\pe_2$ (possibly once clockwise and once counter-clockwise), hence any plane matching on $S$ in $\pe_1$ is also plane in $\pe_2$ and has the same shape.

\item \label{case:2sets-lb-2} Suppose we have already found a compatible matching $M$ consisting of $m\le k-1$ edges.
This leaves at least $n - 2m \ge k^2+1$
points yet unmatched.
The unmatched points are split by the $m$ matching edges into at most $m+1\le k$ subsets, both in $\pe_1$ and in $\pe_2$. Since there are at most $k^2$ different ways to choose one such subset from $\pe_1$ and one from $\pe_2$, there exist 
two yet unmatched points $x$, $y$ that lie in the same subset in $\pe_1$ and in the same subset in $\pe_2$. Hence $xy$ can be added to the matching $M$.

\item This claim is equivalent to Problem 5 given at IMO 2017.\footnote{\url{https://www.imo-official.org/problems/IMO2017SL.pdf}, Problem C4.} 
For completeness we sketch a proof (see~\cref{fig:2convex-lb}):
		split the perimeter of $\pe_2$ into $k$ contiguous blocks $B_1,\dots,B_k$ consisting of $k+1$ points each (that is, block $B_1$ consists of points labeled $\pi(1),\dots,\pi(k+1)$ and so on).
		We aim to draw one matching edge per block. We process points $x_i$ in order $i=1,\dots,n$ in which they appear in $\pe_1$. Once some block, say $B^\star$, contains two processed points, say $x_u$ and $x_v$, we draw edge $x_ux_v$, discard other already processed points and discard other points in $B^\star$.
In this way, any time we draw an edge in some block, we discard at most one point from each other block. Since each block initially contains $k+1$ points, we will eventually draw one edge in each block.
The produced matching contains one edge per block, hence it is non-nested in $\pe_2$. Since points $x_i$ are processed in order $i=1,\dots,n$, the matching is also non-nested in $\pe_1$.

\begin{figure}[h]
\centering
\includegraphics[scale=0.85]{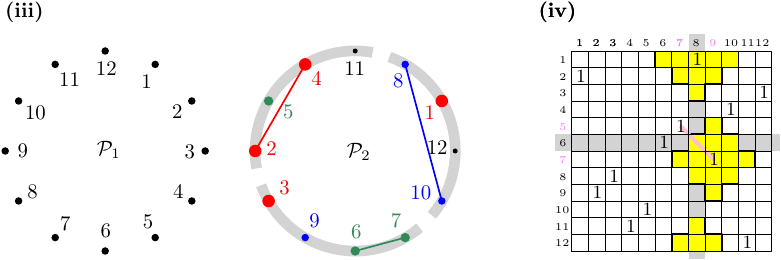}
\caption{Theorem~\ref{thm:2sets-lb}, Claim (iii): Illustration with $n=12$ points and $k=3$ blocks (grey). After drawing an edge we switch the color of processed points (red to green to blue). Claim (iv): The permutation matrix $\Pi$ and two $2$-balls (yellow). A 2-ball centered at $[5,7]$ would intersect a 2-ball at $[7,9]$, so drawing the edge between points labeled 7, 9 forces us to discard at most 2 other points (6 and~8).
}
\label{fig:2convex-lb}
\end{figure}

\item The idea is to find two points $x_i$, $x_j$ that are close to each other in the cyclic order in both $\pe_1$ and $\pe_2$. 
The bound is then established by drawing the edge $x_ix_j$, omitting all points on the shorter arcs of $x_ix_j$ in both $\pe_1$ and $\pe_2$, and proceeding recursively.

Consider the permutation matrix $\Pi$ given by $\pi$, that is, an $n\times n$ matrix such that $\Pi_{i,j}=1$ if $\pi(i)=j$ and 0 otherwise.
Given an integer $r>0$ and a 
cell $\Pi_{i,j}$ containing a digit~1, the \emph{$r$-ball} centered at $\Pi_{i,j}$ is a set $B(\Pi_{i,j},r)=\{\Pi_{u,v} \;:\; |i-u|+|j-v|\le r\}$ of cells whose $L_1$-distance from $\Pi_{i,j}$ is at most $r$, where all indices are considered cyclically modulo~$n$ (see~\cref{fig:2convex-lb}).
Note that an $r$-ball contains $2r^2+2r+1$ cells.

Now suppose $n$ and $r$ satisfy $n\le 2r^2+2r$ and consider $r$-balls centered at all $n$ cells containing a digit~1. The balls in total cover $n\cdot (2r^2+2r+1) >n^2$ cells, hence some two $r$-balls intersect and their centers $\Pi_{i,\pi(i)}$, $\Pi_{j,\pi(j)}$ have $L_1$-distance at most $2r$. This means that the shorter arcs between points labeled $\pi(i)$ and $\pi(j)$ contain, together in both point sets $\pe_1$ and $\pe_2$, at most $2r-2$ other points. Drawing an edge $\pi(i)\pi(j)$ and removing these $2r-2$ other points leaves convex sets in both $\pe_1$ and $\pe_2$ whose convex hulls do not intersect the matched edge $\pi(i)\pi(j)$. 

The rest is induction. The claim holds for $k\in\{1,2\}$. 
Suppose that $k=2r$ is even and that $n= \frac12k^2+k=2r^2+2r$. By the above argument, find a ``short'' edge $x_ix_j$ and remove up to $2r-2$ other points. This leaves $n-2r$ ($<2r^2+2r$) points, so find another edge $x_ux_v$ and remove up to $2r-2$ other points. This leaves $2r^2-2r = 2(r-1)^2+2(r-1)$ points and the induction applies. Last, note that the above shows that having $2r^2$ points implies a $(2r-1)$-matching. Since $2r^2=\lceil \frac{1}{2}(2r-1)^2+2r-1 \rceil$, the case of $k=2r-1$ odd and $n=\lceil \frac{1}{2}(2r-1)^2+2r-1 \rceil$ is also settled. 
		\hfill$\qed$
\end{enumerate}%
\end{proof}

For the remainder of this section, we consider upper bounds on the size of compatible matchings for pairs of convex point sets.

We first describe an explicit construction of two labeled point sets $\pe_{\id}$ and~$\pe_{\pi}$, 
where $n$ is a multiple of 10, the set $\pe_{\id}$ is labeled $(1,2,\ldots,n)$  in clockwise order, 
and the set $\pe_{\pi}$ is labeled $(\pi(1),\pi(2),\ldots,\pi(n))$ in clockwise order, 
by defining a specific permutation $\pi\colon[n]\to[n]$.
We will show that any compatible matching of $\pe_{\id}$ and $\pe_{\pi}$ misses at least $n/5$ of the points.

Our building block for $\pi$ is the permutation $(2,4,1,5,3)$ of five elements. For labeling the $n=5k$ points of $\pe_{\pi}$ (with $k\geq 2$ even) we use the permutation
$\pi = (2,4,1,5,3, \ 7,9,6,10,8, \ \dots, 5(k\!-\!1)\!+\!2,5(k\!-\!1)\!+\!4,5(k\!-\!1)\!+\!1,5(k\!-\!1)\!+\!5,5(k\!-\!1)\!+\!3)$ 
that yields $k$ blocks of 5 points each in both $\pe_1$ and $\pe_2$ (see \cref{fig:5n-construction}). 

\begin{figure}[htb]
\centering
\includegraphics[scale=0.85]{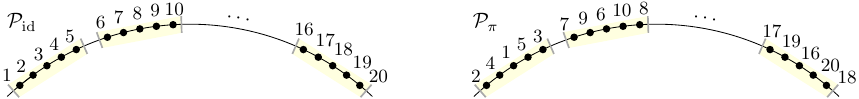}
	\caption{The two labeled point sets $\pe_{\id}$ and $\pe_{\pi}$ for the permutation $\pi$.}
\label{fig:5n-construction}
\end{figure}

\begin{restatable}[Constructive upper bound for two convex sets]{proposition}{twoconstrUB}\label{thm:2sets-c-ub} 
	The largest compatible matching of the two labeled $n$-point sets $\pe_{\id}$ and $\pe_{\pi}$ defined above 
	contains $\frac{2}{5} n$ edges.
\end{restatable}

\begin{proof}
Consider any compatible matching $M$ of $\pe_{\id}$ and $\pe_{\pi}$.
We show that any compatible matching misses at least one point within each of the $k$ blocks.
	This gives $\frac{n}{5}$ unmatched points and thus at most $\frac{2}{5} n$ edges in any compatible matching. 

We classify the edges of any compatible matching $M$ into two types: those that connect two points in one block (that is, edges with both labels in $\{i\!+\!1, \linebreak i\!+\!2,i\!+\!3,i\!+\!4, i\!+\!5\}$ for some $0\le i \le k-1$; we call them \textit{short} edges) and all other edges, which connect two points from different blocks (we call them \emph{long} edges). To show that $M$ misses at least one point of each block~$B$, we distinguish two cases: 
\begin{description}
  \item[Case 1:] \emph{$B$ contains at least one short edge.}

  We show that there is always at least one unmatched point in~$B$.
	W.l.o.g., let $B$ be the block with labels $1, \ldots, 5$. 
		For each of the possible $\binom{5}{2}=10$ short edges, we will verify 
		that if we include it, then we inevitably obtain a point of~$B$ that can not be matched in either $\pe_1$ or $\pe_2$   (see \cref{fig:5n-construction}). 
	We argue by contradiction: suppose $M$ matches all points of~$B$ and~$B$ contains a short edge.
	First, if a short edge cuts off an odd number of points in $B$ in either $\pe_1$ or $\pe_2$, then one of these points is unmatched.
	This is the case for the edges $(1,3), (1,5), (2,4)$ and $(3,5)$ in $\pe_1$
	and for the edges $(1,2), (2,3)$ and $(4,5)$ in~$\pe_2$. So none of these edges can be a short edge in $B$ such that all points of $B$ are matched.
	Second, if $B$ contains the edge $(1,4)$, then the edge $(2,3)$ has to be in~$M$ due to $\pe_1$ and is therefore also a short edge in $B$. This is a contradiction since the edge $(2,3)$ cannot be a short edge in~$B$.
	Third, if $B$ contains the edge $(2,5)$, then the edge $(3,4)$ has to be in~$M$ due to $\pe_1$. This is a contradiction since the edges $(2,5)$ and $(3,4)$ are not compatible due to~$\pe_2$.
	And finally, if~$B$ contains the edge $(3,4)$, then the edge $(1,5)$ has to be in~$M$ due to $\pe_2$ and is therefore also a short edge in~$B$. This is a contradiction since the edge $(1,5)$ cannot be a short edge in~$B$.
	This completes the proof that if~$B$ contains a short edge, then $M$ does not match all points of~$B$.
		
  \item[Case 2:]
	\emph{All five points in~$B$ are matched by a long edge.} 
	
	We argue that, under the assumption that all five points in~$B$ are matched (by a long edge), 
	all those five edges in fact must go to the same block~$B'$, which we then show to be impossible. 
	Consider a pair of numbers $a$, $b$ that lie in the same block whose relative position within that block is different in~$\pe_1$ and in $\pe_2$ (for example, 1 and 2 but not 1 and 3). Suppose~$b$ is matched to~$b'$. Then $a$ has to be matched to a point on the same side of the line $bb'$ as~$a$, in both $\pe_1$ and $\pe_2$. This is impossible unless $a$ is matched to a point in the same block as $b'$ that moreover happens to lie on the correct side of $b'$ in both $\pe_1$ and $\pe_2$. Hence for any such pair~$a$, $b$, the points~$a$ and $b$ are matched to the same block. It remains to notice that $(2,1),(1,4),(4,3),(3,5)$ are all eligible $(a,b)$ pairs, hence all five edges go to the same block $B'$. However, there is only one non-crossing perfect matching of $B$ and $B'$ in $\pe_1$ and we easily check that it is not compatible with $\pe_2$. 
\end{description}

	To see that the bound is tight, note that within each block of $\pe_{\pi}$ we can match the first two points and the next two points.
	This yields a compatible matching of $\pe_{\id}$ and $\pe_{\pi}$ with $2k=\frac{2}{5}n$ edges consisting only of short edges.
	\hfill$\qed$
\end{proof}

The above construction 
yields an upper bound of $\ccm(n) \leq \lceil \frac{2}{5} n \rceil$. However, this bound is not 
tight.
We next show in a probabilistic way that there exists a permutation $\pi\colon [n]\to[n]$ for which the largest compatible matching consists of $k=\bigO(n^{2/3})$ edges. 
In \cref{sec:general}, we will extend this approach to any number of point sets, not necessarily in convex position (\cref{thm:n-ub}).

\begin{restatable}[Probabilistic upper bound for two convex sets]{theorem}{twononconstrUB}\label{thm:2sets-n-ub}
Fix $n$ and let $k=4n^{2/3}$. 
Then two convex sets $P_1$, $P_2$ of $n$ points each can be labeled such that the largest compatible matching consists of fewer than $k$ edges.
\end{restatable}

\begin{proof} 
Let $\pe_1$ be $P_1$ with labeling $(1,2,\dots,n)$ in clockwise order and let $P_2$ not yet be labeled.
The idea for this proof is that for large $n$ there are more ways to label $P_2$ than there are ways to draw a compatible $k$-matching.

For any $k\le n$, let $f(k)$ be the number of plane $k$-matchings of $P_i$, $i \in \{1,2\}$ 
	(that is, matchings leaving $n-2k$ points unmatched). 
	As there are $\binom{n}{2k}$ ways to select the $2k$ points to be matched and
	the number of plane perfect matching on those points is $\frac{1}{k+1}\binom{2k}{k}$ (the $k$-th Catalan number),
 we obtain
	\[f(k) = \binom{n}{2k}\cdot \frac{1}{k+1}\binom{2k}{k} \le \frac{n!}{(n-2k)!\cdot k!\cdot k!} \ .\]

Given two plane $k$-matchings, one of $\pe_1$ and one of $P_2$, 
	there are exactly $g(k) = (n-2k)!\cdot k!\cdot 2^{k}$
	labelings of $P_2$ for which those two matchings constitute a compatible $k$-matching:
	there are $(n-2k)!$ ways to label the unmatched points of~$P_2$, $k!$ ways to pair up the matching edges and $2^k$ ways to label their endpoints.

	Therefore, $(f(k))^2 g(k)$ is an upper bound for the number of labelings $\pi$ of $P_2$ such that there is a compatible $k$-matching for $\pe_1$ and $\pe_2$ ($P_2$ with labeling $\pi$).
	On the other hand, there are $n!$ labelings of $P_2$ in total.

	Our goal is to show that 
	$(f(k))^2\cdot g(k) < n! $.
	If we succeed, then there exists a labeling $\pi$ of $P_2$ such that there is no compatible $k$-matching for $\pe_1$ and $\pe_2$ ($P_2$ with labeling $\pi$).
	Canceling some of the factorials and using standard bounds $(n/\e )^n<n!<n^n$ on the remaining ones (where $\e$ 
	denotes Euler's number), we obtain 
	\begin{eqnarray*} \frac{(f(k))^2\cdot g(k)}{n!} \le \frac{n! \cdot 2^k}{(n-2k)!\cdot(k!)^3} \le \frac{n^{2k}\cdot 2^k}{(k/\e )^{3k}}& = & \left( \frac{2\e ^3 n^2}{k^3} \right)^k. \end{eqnarray*}
	For $k\geq 4 n^{2/3}$, the above expression is less than one (we have $2\e ^3<4^3$), which completes the proof. 
	\hfill$\qed$
\end{proof}

\section{Generalized and Multiple Sets}
\label{sec:general}

In this section we generalize our results in two ways, by considering point sets in general position and more than two sets. 
We again start with lower bounds. 
\Cref{thm:lb}, which is a generalization of the second result of \Cref{thm:2sets-lb}, implies that for any $\sets$-tuple of point sets $P_1,\dots,P_{\sets}$ we have $\cmk(n;P_1,\dots,P_{\sets})=\Omega(n^{1/\sets})$.

\begin{theorem}[Lower bound for multiple sets]
\label{thm:lb}
Let $\pe_1, \pe_2,\dots,\pe_{\sets}$ be labeled sets of $n$ points each.
If $n\ge k^{\sets}+2k-1$, then any maximal compatible matching consists of at least $k$ edges.
\end{theorem}
\begin{proof} We extend the idea from the proof of \cref{thm:2sets-lb}, part~(\ref{case:2sets-lb-2}): 
 suppose we have already found a compatible matching $M$ consisting of $m\le k-1$  edges.
This leaves at least $k^{\sets}+2k-1 - m \ge k^{\sets}+2k-1 - 2(k-1) = k^{\sets}+1$ points yet unmatched.
Imagine the $\sets$ point sets live in $\sets$ different planes. We process the $m$ matching edges one by one. When an edge is processed, we extend it along its line in both directions until it hits another matching edge or an extension of a previously processed edge (in all $\sets$ planes). In this way, the $m$ lines partition each plane into $m+1\le k$ convex regions. By simple counting ($k^{\sets}+1>k^{\sets}$), there exist two yet unmatched points $x$, $y$ that lie in the same region in each of the $\sets$ planes. 
Hence $xy$ can be added to the matching $M$. 
\hfill$\qed$
\end{proof}

Regarding upper bounds, the following theorem 
implies that for any fixed $\sets$ and any \mbox{$\sets$-tuple} of 
point sets $P_1,\dots,P_{\sets}$, we have $\cmk(n;P_1,\dots,P_{\sets})= \bigO(n^{2/(\sets+1)})$.

\begin{restatable}[Probabilistic upper bound for multiple sets]{theorem}{nonconstUBmultple}\label{thm:n-ub}
Fix $n$ and $\sets$ and let $k=125\cdot n^{2/(\sets+1)}$. Then any $\sets$ sets $P_1,\dots,P_{\sets}$ of $n$ points each, where each $P_i$ is in general position, can be labeled such that the largest compatible matching consists of fewer than $k$ edges.
\end{restatable}

This theorem can be proven by extending the idea from the proof of \cref{thm:2sets-n-ub} and combining results of Sharir, Sheffer and Welzl~\cite{SHARIR2013777} and Sharir and Sheffer~\cite{SharirS11} on the number of triangulations and plane perfect matchings.

\begin{proof}
	Let $\pe_1$ be $P_1$ with labeling $(1,2,\dots,n)$ in clockwise order and suppose that the remaining $\sets-1$ sets are not yet labeled.
	Let $k\le n$ and let $f_i(k)$, $1 \le i \le \sets$, be the number of $k$-matchings of $P_i$.
Sharir, Sheffer and Welzl showed in~\cite{SHARIR2013777} that the number of plane perfect matchings of any set $P$ of $k$ points in general position is at most $8\cdot(3/2)^{(k/4)}\cdot tr(P)$, where $tr(P)$ denotes the number of triangulations of $P$.
Sharir and Sheffer also showed in~\cite{SharirS11} that the number of triangulations of $P$ 
is at most $30^k$.
This implies that there are also at most $8\cdot33.21^k$ different perfect matchings of $P$. 
By counting this upper bound for every possible $k$-point subset of $P_i$, we obtain
	\begin{eqnarray*} f_i(k)  \le  &  \binom{n}{2k}\cdot 8 \cdot 33.21^k. \end{eqnarray*}

Next, consider an $\sets$-tuple $(M_1,\dots,M_{\sets})$ of matchings on $P_1,\dots,P_{\sets}$, respectively, consisting of $k$ edges each. Notice that any such $\sets$-tuple forms a compatible matching for 
	\begin{eqnarray*} g(k) & = & \left( (n-2k)!\cdot k!\cdot 2^{k} \right)^{\sets-1} \end{eqnarray*}
combinations of labelings for $P_2,\dots,P_{\sets}$.
On the other hand, there are $( n!)^{\sets-1}$ such combinations of labelings for $P_2,\ldots,P_\sets$ in total.
It suffices to show that 
	\begin{eqnarray*}\label{eq:prob_many} \left(\prod_{i=1}^{\sets} f_i(k)\right) \cdot g(k)  & < & ( n!)^{\sets-1} \end{eqnarray*}
for guaranteeing the existence of a combination of labelings for $P_2,\ldots,P_\sets$ such that there is no compatible $k$-matching for the resulting labeled sets $\pe_1, \ldots, \pe_\sets$.
As before, we expand the binomials into factorials, cancel some of them and use standard bounds $(n/e)^n<n!<n^n$ on the remaining ones (where $e$ again denotes Euler's number) to obtain
\begin{eqnarray*}
	\frac{\left(\prod_{i=1}^{\sets} f_i(k)\right) \cdot g(k) }{ ( n!)^{\sets-1}}
	&  =  &  \frac{ n!\cdot (8\cdot {33.21}^k)^{\sets} (k!)^{\sets-1}(2^k)^{\sets-1}}{(n-2k)! ((2k)!)^{\sets}} \\
	& \le &  8^{{\sets}} \frac{ n^{2k} \cdot (66.42^{\sets})^k (k^{\sets-1})^k (e^{2\sets})^k }{ ((2k)^{\sets})^{2k}  } 
 \le 8^{{\sets}} \left[(16.605 \cdot e^2)^{\sets} \cdot \frac{n^2}{k^{\sets+1}}  \right]^k.
\end{eqnarray*}
When $k=125\cdot n^{2/({\sets}+1)}$, then also $k\ge 125$ holds.
Further, the expression inside  the brackets is less than $\frac{1}{1.018^{\sets}}$  (we have $16.605\cdot e^2<\frac{125}{1.018}$).
Since $8<1.018^{125} \le 1.018^k$, this completes the proof.
\hfill$\qed$
\end{proof}

We remark that the upper bound of $33.21^k$ for the number of plane perfect matchings of any set of $k$ points in the plane in general position is by far not tight. 
Actually, Sharir and Welzl~\cite{SharirW06} showed that this number can be bounded by $O(10.05^k)$.
However, for the above proof, we require an explicit upper bound that holds for any value of $k\ge 125$ and hence we did not use this result.

\section{Forcing a single-edge compatible matching}
\label{sec:force}

In this section we consider the following question:
Given an unlabeled point set $P$ with $n$ points,
is there an integer $\sets$ such that there exist $\sets$ labelings of $P$ for which every compatible matching has at most one edge?
If $\sets$ exists, we denote as $\force(n;P)$ the minimum number $\sets$ of copies of $P$ such that $\cmk(n;P,\dots,P)=1$ (where $P$ appears $\force(n;P)$ times).
Otherwise, we set $\force(n;P)=\infty$.
In other words, we are asking for the existence (and minimal number) of labelings of the set $P$ so that any pair of labeled edges crosses for at least one labeling. 
We remark that, again, the size $n$ of the point sets is 
included in the notation only for the sake of clarity. 

Note that $\force(n;P) = \infty$ if and only if the straight-line drawing of $K_n$ on~$P$ does not contain any crossing.
Hence $\force(n;P)$ is finite for any set $P$ of $n\geq 5$ points.
More specifically, if the straight-line drawing of $K_n$ on $P$ contains at least one crossing, then $\force(n;P)$ is at most $3{n \choose 4}=\bigO(n^4)$.
This bound is due to the fact that, if we focus on one crossing edge pair, then there are $3{n \choose 4}$ possible labelings for these two edges.
If the straight-line drawing of $K_n$ on $P$ contains exactly one crossing, then $3{n \choose 4}$ is tight:
In this case, each pair of labeled edges must be mapped to the unique crossing edge pair, as otherwise, that pair of labeled edges would be compatible. 
Hence, any point set $P_4$ of 4 points in convex position has $\force(4;P_4)=3$
and any point set $P_5$ of 5 points with triangular convex hull has $\force(5;R_5)=15$.

We first focus on upper bounds and on the case when $P$ is in convex position. 
We denote by $\cforce(n)$  the minimum number of copies of a convex set with $n$ points that
need to be labeled so that the largest compatible matching consists of only a single edge.

Let $b(n)=\lceil \log_2 n \rceil$, which is the number of bits that are needed to represent the labels 1 to $n$. We construct a family of $\frac{3}{2}b(n)^2$ labeled convex $n$-point sets such that all pairs of edges cross in at least one set. 
First consider three labeled convex point sets, which are obtained by partitioning the set of labels into four blocks $A$, $B$, $C$, $D$, and combining those blocks in different orders and orientations as depicted in \cref{fig:three_drawings}.
The order within a block is arbitrary, but identical for all three sets (up to reflection; those block orientations are indicated by arrows).

\begin{figure}[ht]
	\centering
	\begin{subfigure}{.31\textwidth}
		\centering
		\includegraphics[scale=0.7,page=2]{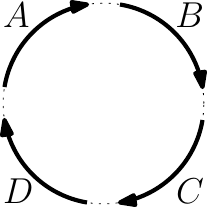}
		\caption{}
		\label{fig:three_drawings1}
	\end{subfigure}
	\begin{subfigure}{.31\textwidth}
		\centering
		\includegraphics[scale=0.7,page=3]{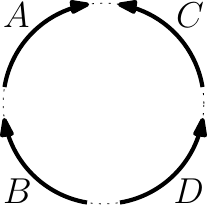}
		\caption{}
		\label{fig:three_drawings2}  
	\end{subfigure}
	\begin{subfigure}{.31\textwidth}
		\centering
		\includegraphics[scale=0.7,page=4]{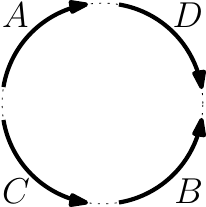}
		\caption{}
		\label{fig:three_drawings3}  
	\end{subfigure}
	\caption{Three labeled point sets obtained from different orders and orientations of four blocks $A$, $B$, $C$, and $D$.
	}
	\label{fig:three_drawings}
\end{figure}

\setcounter{lemmathreesets}{\value{lemma}}
\begin{restatable}{lemma}{constUBcforce}\label{lem:ub_force_preparation}
	Consider three convex point sets $\pe_1$, $\pe_2$, and $\pe_3$ that are labeled as in \cref{fig:three_drawings} for some partition $A$, $B$, $C$, and $D$ of their label set.
	Then any pair of independent\footnote{Two edges are independent if they do not share an endpoint.} edges, where none of them has both labels in one of the blocks $A$, $B$, $C$, and $D$, forms a crossing in at least one of $\pe_1$, $\pe_2$, and $\pe_3$.
\end{restatable}

\begin{proof}
	We consider three cases, depending on the number $x$ of blocks containing endpoints of the edges $e$ and $f$.
	\begin{description}
		\item[Case $x=2$:] Then $e$ and $f$ are spanned by the same two subsets. As any pair of blocks shows up in the same and in inverse orientation in at least one of the three drawings, this guarantees a crossing. For example, let the two subsets be $A$ and $B$. They have the same orientation in~\cref{fig:three_drawings1}  and~\ref{fig:three_drawings2}, but inverse orientation in~\cref{fig:three_drawings3}. An analogous property holds for the remaining five combinations.
		\item[Case $x=3$:]  W.l.o.g.\ $e$ and $f$ have their starting point in the same subset, but the endpoints in different subsets. There are 12 possible configurations of one common and two disjoint subsets, and it is straightforward to check that each situation shows up in both possible orientations with respect to the common set. For example, let the common set be $A$, and the other sets $C$ and $D$. The orientation of $A$ is the same in all three drawings, but the order of $C$ and $D$ is 
		inverted in~\cref{fig:three_drawings1} and~\ref{fig:three_drawings3}. If the common set is $B$ and the two other sets are again $C$ and $D$, then order of the three sets is the same in all three drawings, but the orientation of the common set $B$ is inverted in~\cref{fig:three_drawings1} and~\ref{fig:three_drawings3}. Thus, in both cases a crossing is guaranteed.
		\item[Case $x=4$:] In this case the orientation of the subsets is not relevant. There are only three possible combinations of such edges ($A-B$ with $C-D$, $A-C$ with $B-D$, and $A-D$ with $B-C$) and the three drawings cover one case each. \hfill$\qed$ 
	\end{description}
\end{proof}

	We next identify a small number of 4-partitions of the label set $\{1,2,\ldots,n\}$ such that 
	each edge pair fulfills the condition of \cref{lem:ub_force_preparation} in at least one of the partitions (when the four subsets form blocks).
	This yields the following constructive upper bound for $\cforce(n)$.

\begin{proposition}[Constructive upper bound on {\boldmath $\cforce(n)$}]\label{lem:partitions}
	For any $n\!\geq\!4$ and for $b(n)=\lceil \log_2 n\rceil$, we can define $3{b(n) \choose 2}$ labeled convex sets of $n$ points such that the largest matching compatible to all of them consists of a single edge. 
\end{proposition}

\begin{proof}
	Given a convex set of $n$ points, we construct $b(n) \choose 2$ 4-partitions of the labels and use each such partition to obtain three labeled point sets as depicted in \cref{fig:three_drawings}. 
	For any two bit positions $i,j$,  $0 \leq i \neq j < b(n)$, of the labels,
	partition the label set $\{1,2,\ldots,n\}$ so that $A$ contains all labels where those two bits are zero, $B$ those where the bits are zero-one, $C$ those with one-zero, and finally $D$ the ones with both one. This gives $b(n) \choose 2$ different partitions. 
	
	Now consider two arbitrary edges $e$ and $f$. Then there is a bit position in which the two endpoints of $e$ have different values, and the same is true for $f$. Let $i$ and $j$, respectively, be those positions. If this would give $i=j$, then choose $j$ arbitrarily but not equal to $i$. By \cref{lem:ub_force_preparation}, the edges $e$ and $f$ cross in one of the three labeled point sets for the partition generated for $i$ and $j$. 
\hfill$\qed$
\end{proof}

The upper bound $\bigO(\log^2 n)$ of $\cforce(n)$ from \cref{lem:partitions} is constructive but it is not asymptotically tight.
Next we present a probabilistic argument which shows that we actually have $\force(n;P)=\bigO(\log n)$ for any point set $P$ of $n\ge 5$ points.

\begin{lemma}[Probabilistic upper bound on {\boldmath $\force(n;P)$}]\label{lem:upperforce}
	Given a set ${P}$ of $n\geq 5$ points in general position, there exists a constant $c_P\ge15/14$ such that
	$\force(n;P)\leq \log_{c_P}(3{n \choose 4})=\bigO(\log n)$. 
\end{lemma}
\begin{proof}
Fix $P$ and let $\alpha_P\in(0,1)$ be the proportion of 4-tuples of points in ${P}$ that are in convex position.
Note that since any 5-tuple of points contains at least one 4-tuple in convex position, we have $\alpha_P\ge 1/5$ (here we use $n\ge 5$).

There are $r=3\binom{n}{4}$ pairs of non-incident edges.
Fix one of them, say $ac$ and~$bd$.
Note that when ${P}$ is labeled uniformly at random, the edges $ac$, $bd$ intersect with constant probability $\alpha_P/3$:
indeed, the edges intersect if their 4 endpoints form a convex quadrilateral and the points $a$, $b$, $c$, $d$ lie on its perimeter in two out of the six possible cyclic orders.

Now set $c_P=(1-\alpha_P/3)^{-1}\ge 15/14$ and consider $\sets>\log_{c_P}(r)$ copies of ${P}$ labeled independently and uniformly at random.
A pair of edges without common vertex is compatible (with respect to the $\sets$ labeled point sets) if it is non-crossing in all $\sets$ point sets.
As the labelings are chosen independently and uniformly at random, any fixed pair of edges is compatible with probability $\rho=(1-\alpha_P/3)^{\sets}<1/r$.
By linearity of expectation, the expected number of compatible pairs of edges is $r\cdot \rho <1$.
Therefore there exists a labeling of the $\sets$ point sets for which no pair of edges is compatible.
In such a labeling, the largest compatible matching consists of a single edge.
\hfill$\qed$
\end{proof}

We remark that for a fixed set $P$, one can often obtain a better lower bound on the parameter $\alpha_P$ used in the proof and thus a better lower bound on $c_P$, which then gives a constant factor improvement on $\force(n;P)$.
Specifically, finding the constant $\alpha := \lim_{n\to\infty} \inf_{P,|P|\ge n}\{\alpha_P\}$
is a topic of high relevance in connection with the rectilinear crossing number of the complete graph, see~\cite{crsurvey} for a nice survey of this area.
The currently best known bounds are $0.37997256 < \alpha < 0.38044919$~\cite{crlowerbound,crupperbound}.
Moreover, when $P$ is in convex position we have $\alpha_P=1$ and thus the above proof implies $\cforce(n)\leq \log_{3/2}(3{n \choose 4})$.
On the other hand, none of these observations leads to an asymptotic improvement on the upper bound of $\force(n;P)$ or $\cforce(n)$. In the following we show that any such asymptotic improvement is in fact impossible.

\begin{lemma}[Lower bound on {\boldmath $\force(n;P)$}] \label{lem:lowerforce}
	Fix $k\ge 1$ and let $P$ be any set of $n=2^k+3$ points in general position.
	Then $\force(n;P) \ge k+2=\Omega(\log n)$.
\end{lemma}
\begin{proof}
	We use a similar argument as the one used in \cref{thm:lb}.
	Denote by $\pe_1,\dots,\pe_{k+1}$ any $k+1$ labeled copies of $P$. 
	Take an arbitrary edge $ab$ on the convex hull of $\pe_{k+1}$. The line containing $ab$ divides each of $\pe_1,\dots,\pe_k$ into two parts (one possibly empty).
	Since there are $n-2=2^k+1>2^k$ unmatched points, there exist two points, say $x$, $y$, that lie in the same part, for each $i=1,\dots,k$.
	Thus the two edges $xy$ and $ab$ form a compatible 2-matching implying that $\force(n;P) \ge k+2$.
\hfill$\qed$
\end{proof}

Combining upper and lower bounds for $\force(n;P)$ from \cref{lem:lowerforce,lem:upperforce}, we obtain the following conclusion: 

\begin{theorem}\label{cor:forcetightbound}
	For every set $P$ of $n \geq 5$ points in general position, it holds that 
	\begin{eqnarray*} \force(n;P) & = & \Theta(\log n). \end{eqnarray*}
\end{theorem}

\section{Conclusion and open problems}

In this work we studied the size of compatible matchings.
A natural open problem is the computational complexity of finding compatible matchings of a certain size or even deciding their existence:  
		How fast can we decide if two (or more) given (general or convex) labeled point sets have a perfect compatible matching (or a compatible matching of size $k$)?

We showed that for any $\sets$ labeled $n$-point sets, there is always a compatible matching of size $\Omega(n^{1/\sets})$. 
On the other hand, for any $\sets$ unlabeled $n$-point sets, there exist labelings of these $\sets$ point sets such that the largest compatible matching has size $\bigO(n^{2/(\sets +1)})$.
Even for only two sets these bounds are currently the best ones known.
It would be interesting to close this gap, even if this only hold for a special point set.

Furthermore, for two convex point sets, we only have an explicit labeling such that the largest compatible matching has $2n/5$ matching edges.
This leads to the following open problems:
Can we give a family of explicit labeled convex point sets, such that the largest compatible matching has size $\bigO(n^{1-\epsilon})$?
Furthermore, can we give a labeling of $\sets$ point sets that allows only small compatible matchings?
Note that the second problem is also open if we only look at convex point sets.
It would also be nice to have a construction that matches the probabilistic bound, even if the construction only works for special point sets.

Further we studied how many labeled copies of a point set needed such that any compatible matching only contains one edge.
We showed that $\Theta(\log n)$ point sets are sufficient and that this bound is tight.
We constructed a labeling of roughly $\frac{3}{2} \log^2 n$ convex labeled point sets that obtain this property.
This leads to the following question:
Can we construct a family of $\Theta(\log n)$ (maybe convex) labeled point sets such that any compatible matching only contains one edge?

Finally, as a game version of this problem, consider the following game:
Two players alternately add an edge which must neither cross nor be incident to any previously added edge. The last player who is able to add such an edge wins. 
It is not hard to see that for a single set of points in convex position, this is the well-known game Dawson's Kayles, see e.g.~\cite{winningways}. This game can be perfectly solved using the nimber theory developed by Sprague-Grundy, see also~\cite{winningways} for a nice introduction to the area. An interesting generalization of Dawson's Kayles occurs
when the players use two (or more) labeled (convex) point sets and add compatible edges.

\bibliographystyle{splncs04}
\bibliography{references}

\end{document}